\documentclass{scrartcl}

\usepackage{amsmath, amsthm, amssymb}
\usepackage{hyperref}

\usepackage[numbers]{natbib}
\usepackage{tikz}
\usetikzlibrary{snakes}
\tikzstyle{labeledNode}=[circle, inner sep = 0.1em, minimum size = 1.2em, draw]
\tikzstyle{normalNode}=[circle, fill=black, draw]
\tikzstyle{dualNode}=[draw]
\tikzstyle{smallNode}=[draw, circle, fill, scale=0.5]
\tikzstyle{smallLabeledNode}=[draw, circle, scale=0.5]
\tikzstyle{normalEdge}=[thick, >=stealth]
\tikzstyle{dualEdge}=[thick, >=stealth, dashed]

\usepackage{bbm}

\theoremstyle{plain}
\newtheorem{theorem}{Theorem}[section]
\newtheorem{corollary}[theorem]{Corollary}
\newtheorem{lemma}[theorem]{Lemma}

\theoremstyle{definition}
\newtheorem{definition}[theorem]{Definition}
\newtheorem{theorem_definition}[theorem]{Theorem and Definition}
\newtheorem{remark}[theorem]{Remark}
\newtheorem*{notation}{Notation}
\newtheorem*{assumption}{Assumption}

\newenvironment{prooflist}{\begin{list}{\labelitemi}{\leftmargin=2em \itemindent=-1em}}{\end{list}}

\newcommand{\vecspan}{\operatorname{span}}

\newcommand{\cyclespace}[1]{\mathcal{S}_{\text{cycle}}(#1)}
\newcommand{\darts}[1]{\overleftrightarrow{#1}}
\newcommand{\setof}{\mathcal{D}}

\newcommand{\arc}[1]{\overrightarrow{#1}}
\newcommand{\antiarc}[1]{\overleftarrow{#1}}
\newcommand{\rev}[1]{\operatorname{rev}(#1)}
\newcommand{\tail}{\operatorname{tail}}
\newcommand{\head}{\operatorname{head}}
\newcommand{\leftf}{\operatorname{left}}
\newcommand{\rightf}{\operatorname{right}}
\newcommand{\bigo}{\mathcal{O}}

\newcommand{\rightof}{\preceq}
\newcommand{\leftof}{\succeq}

\begin{document}
 \title{Lattices and maximum flow algorithms\\ in planar graphs}
 \date{}
 \author{Jannik Matuschke\thanks{Technische Universit\"at Berlin, Institut f\"ur Mathematik, Stra{\ss}e des 17.~Juni 136, 10623 Berlin, Germany. e-mail: {\tt \{matuschke,peis\}@math.tu-berlin.de}. This work was supported by Berlin Mathematical School. An extended abstract has appeared in \href{http://www.springerlink.com/content/r6l436056k1j0351/}{\emph{36th International Conference on Graph-theoretic Concepts in Computer Science, WG 2010}}.} \and Britta Peis\footnotemark[1]}
 \maketitle
 
 \begin{abstract}
  \textbf{Abstract.} We show that the left/right relation on the set of $s$-$t$-paths of a plane graph induces a so-called submodular lattice. If the embedding of the graph is $s$-$t$-planar, this lattice is even consecutive. This implies that Ford and Fulkerson's uppermost path algorithm for maximum flow in such graphs is indeed a special case of a two-phase greedy algorithm on lattice polyhedra. We also show that the properties submodularity and consecutivity cannot be achieved simultaneously by any partial order on the paths if the graph is planar but not $s$-$t$-planar, thus providing a characterization of this class of graphs.
 \end{abstract}
 
\section{Introduction and preliminaries}

The special case of flows in planar graphs has always played a significant role in network flow theory. The predecessor of Ford and Fulkerson's well-known path augmenting algorithm -- and actually the first combinatorial flow algorithm at all -- was a special version for $s$-$t$-planar networks, i.e., those networks where $s$ and $t$ can be embedded adjacent to the infinite face~\cite{FF56}. The basic idea of this \emph{uppermost path algorithm} is to iteratively augment flow along the ``uppermost'' non-saturated $s$-$t$-path in the planar embedding of the network. In 2006, Borradaile and Klein~\cite{BK06} established an intuitive generalization of this algorithm to arbitrary planar graphs, which relies on a partial order on the set of $s$-$t$-paths in the graph, called the \emph{left/right relation}.

For the special case of $s$-$t$-planar graphs, Hassin~\cite{Ha81} showed that a maximum flow can be computed by a single shortest path computation in the dual. Combining this with the linear time shortest path algorithm in planar graphs of Henzinger et al.~\cite{HK97} yields a linear time algorithm for maximum flow in $s$-$t$-planar graphs. Recently, Erickson~\cite{Er10} showed that also the algorithm of Borradaile and Klein corresponds to a sequence of parametric shortest path computations in the dual graph, providing a simplified analysis of this algorithm. More details on the history of maximum flow in planar graphs can be found in Borradaile and Klein's article \cite{BK06}.

Another area of combinatorial optimization, which has so far been unrelated to planar flow computations, is the optimization on lattice structures. In 1978, Hoffman and Schwartz introduced the notion of \emph{lattice polyhedra}~\cite{HS78}, a generalization of Edmond's polymatroids based on lattices, and proved total dual integrality of the corresponding inequality systems if certain additional properties hold. Later, several variants of two-phase greedy algorithms were developed, e.g., by Kornblum~\cite{Ko78}, Frank~\cite{Fr99}, and Faigle and Peis~\cite{FP08}, to solve quite general classes of linear programs on these polyhedra efficiently.

\paragraph{Our results} In this paper, we connect these two fields of research by showing that the left/right relation induces a lattice on the set of simple $s$-$t$-paths in a planar graph. If the network is $s$-$t$-planar, this lattice fulfills the two main properties required in Hoffman and Schwartz' framework, called \emph{submodularity} and \emph{consecutivity}. Our result implies that the uppermost path algorithm of Ford and Fulkerson is a special case of the two-phase greedy algorithm on lattice polyhedra, which, even more, can solve a variant of the flow problem with supermodular and monotone weights on the paths.
However, the case of general planar graphs, i.e., not necessarily $s$-$t$-planar graphs turns out to be much more involved. In fact, we will characterize $s$-$t$-planar graphs as the only class of planar graphs that can be equipped with a lattice on the set of paths that is consecutive and submodular at the same time.

\paragraph{Outline} In the remainder of this section we will define lattice polyhedra (Subsection \ref{sub_latticepoly}) and introduce the basic notions of graph structures (Subsection \ref{sub_graphs}) and the left/right relation (Subsection \ref{sub_leftright}) we need to present our results. 
We then will discuss the left/right relation in $s$-$t$-planar graphs and provide an intuitive characterization for the relation in this class of graphs, which leads to the insight that the relation induces a submodular and consecutive lattice on such graphs (Section \ref{sec_stplanar}). In Section \ref{sec_general}, we discuss the case of general planar graphs and outline a considerably more involved proof to show that the left/right relation induces a submodular lattice in the general case. Finally, in Section \ref{sec_characterization} we show that consecutivity and submodularity cannot be achieved at the same time by any partial order in the non-$s$-$t$-planar case.

\subsection{Lattice polyhedra}
\label{sub_latticepoly}

Our interest in lattices is motivated by a two-phase greedy algorithm that can solve a primal/dual pair of very general linear programming problems on so-called lattice polyhedra, which have first been introduced by Hoffman and Schwartz \cite{HS78}. More precisely, we are given a finite set $E$, a set system $\mathcal{L} \subseteq 2^{E}$, and two vectors $c \in \mathbb{R}^{E}$, $r \in \mathbb{R}^{\mathcal{L}}$, and consider the covering problem

\[(C)\quad \min\ \left\{\sum_{e \in E} c(e) x(e) ~:~ x \in \mathbb{R}^{E}_{+},~ \sum_{e \in S} x(e) \geq r(S) ~ \forall S \in \mathcal{L}\right\}\]
	and its dual, the packing problem
	\[(P)\quad \max\ \left\{\sum_{S \in \mathcal{L}} r(S) y(S) ~:~ y \in \mathbb{R}^{\mathcal{L}}_{+},~ \sum_{S \in \mathcal{L} : e \in S} y(S) \leq c(e) ~ \forall e \in E\right\}.\]
	
Observe that the packing problem $(P)$ corresponds to an ordinary max flow problem if $\mathcal{L}$ is the set of $s$-$t$-paths of a given network and $r \equiv 1$. Before we can state Hoffman and Schwartz' main result, we need to introduce some definitions.

\begin{definition}[Lattices, submodularity, consecutivity]
\label{def_lattice}
	Let $E$ be a finite set, $\mathcal{L} \subseteq 2^{E}$ and $\preceq$ be a partial order on $\mathcal{L}$. The pair $(\mathcal{L}, \preceq)$ is a \emph{lattice} if for all $S, T \in \mathcal{L}$ the following conditions are fulfilled.
	\begin{itemize}
	 \item $\{L \in \mathcal{L} : L \preceq S, L \preceq T \}$ has a unique maximum element $S \wedge T$, called \emph{meet}.
	 \item $\{U \in \mathcal{L} : U \succeq S, U \succeq T \}$ has a unique minimum element $S \vee T$, called \emph{join}.
	\end{itemize}
	A function $r : \mathcal{L} \rightarrow \mathbb{R}$ is \emph{submodular} if $r(S \wedge T) + r(S \vee T) \leq r(S) + r(T)$ for all $S, T \in \mathcal{L}$. It is \emph{supermodular}, if $r(S \wedge T) + r(S \vee T) \geq r(S) + r(T)$ for all $S, T \in \mathcal{L}$.\\ A lattice $\mathcal{L}$ is \emph{submodular}, if $(S \wedge T) \cap (S \vee T) \subseteq S \cap T \text{ and } (S \wedge T) \cup (S \vee T) \subseteq S \cup T$ for all $S, T \in \mathcal{L}$.\footnote{Submodularity of lattices is connected to submodularity of functions in the following way: A lattice is sub\-modular if and only if all functions of the type $f(S) := \sum_{e \in S} x(e)$ for some vector $x \in \mathbb{R}_{+}^{E}$ are submodular.}  It is \emph{consecutive} if $S \cap U \subseteq T$ for all $S, T, U \in \mathcal{L}$ with $S \preceq T \preceq U$.
\end{definition}

Hoffman and Schwartz showed that the inequality system defining $(C)$ is totally dual integral if $\mathcal{L}$ is a submodular and consecutive lattice and $r$ is supermodular w.r.t.\ the lattice. In this case, the correspondig polyhedron is called \emph{lattice polyhedron}. If $r$ is furthermore monotone and both the lattice and $r$ are polynomially computable in the sense that the maximum element of any restricted sublattice of $\mathcal{L}$ can be found, then there exists a polynomial time algorithm \cite{FP08}.

\subsection{Graphs}
\label{sub_graphs}

Our results in Sections \ref{sec_stplanar} and \ref{sec_general} will be valid for both directed and undirected graphs.
We will assume that we are given a directed graph $G = (V, E)$ (if the graph is undirected, we can direct it arbitrarily), but we will allow paths to use all edges in arbitrary direction.\footnote{This helps to streamline the proofs, the resulting lattice can be restricted to directed paths later on by removing all paths that use backward darts. Note that removing elements from the ground set preserves submodular lattice structures.} For this purpose it will be convenient to equip every edge $e \in E$ with two antiparallel \emph{darts}, a forward dart $\arc{e} := (e, 1)$ pointing in the same direction as the edge and a backward dart $\antiarc{e} := (e, -1)$ pointing in the opposite direction. 

\begin{definition}
 We define $\darts{E} := E \times \{1, -1\}$ to be the set of all darts. For a dart $(e, i)$ we use $\rev{(e, i)} := (e, -i)$ to refer to its reverse. For $e = (v, w)$, we let $\tail(\arc{e}) = v = \head(\antiarc{e})$ and $\head(\arc{e}) = w = \tail(\antiarc{e})$. For $D \subseteq \darts{E}$, we define $E(D) := \{e \in E : \arc{e} \in D \text{ or } \antiarc{e} \in D\}$. We use $G[\tilde{E}]$ to refer to the subgraph that only contains the edges $\tilde{E} \subseteq E$.
\end{definition}

The basic notions of paths, cycles, and cuts are defined in the natural way except that all of these objects consist of darts rather than edges. 

\begin{definition}[Walk, path, cycle]
	An \emph{$x$-$y$-walk} is a non-empty sequence of darts $d_{1}, \ldots, d_{k}$ such that $\head(d_{i}) = \tail(d_{i+1})$ for $i \in \{1, \ldots, k - 1\}$ and $x = \tail(d_{1})$ and $y = \head(d_{k})$. If for all darts of an $x$-$y$-walk the underlying edges are pairwise distinct, then the walk is called \emph{$x$-$y$-path} for $x \neq y$ or \emph{cycle} if $x = y$. A path or cycle is called \emph{simple} if the heads of all its darts are pairwise distinct.
\end{definition}

\begin{definition}[Cut]
\label{def_cut}
	A \emph{cut} is a non-empty set of darts $C \subseteq \darts{E}$ such that there is a set of vertices $S \subset V$ with $C = \Gamma^{+}(S) := \{d \in \darts{E} : \tail(d) \in S,~\head(d) \notin S\}$. The cut $C$ is \emph{simple}, if $S$ and $V \setminus S$ are the connected components of $G[E \setminus E(C)]$.
\end{definition}

A \emph{planar graph} is a graph that can be drawn on (or embedded in) the plane without any two edges intersecting. A graph together with such an embedding is called \emph{plane graph}. The embedding partitions the plane into regions that are bordered by the edges. These regions are called \emph{faces} and can be used to define the dual graph $G^{\ast}$ as follows. The vertex set $V^{\ast}$ of $G^{\ast}$ is the set of all faces. For every edge in $G$, we introduce a corresponding edge in $G^{\ast}$ that connects the faces that are separated by this primal edge, going from right to left. We refer to the faces left and right of a dart $d \in \darts{E}$ by $\leftf(d)$ and $\rightf(d)$, respectively. The face surrounding the drawing is called the \emph{infinite face} $f_{\infty}$.\footnote{Note  that the infinite face can be chosen abritrarily.} We will furthermore make use of the following property of planar graphs. 

\begin{theorem}[Cycle/cut duality]
 $C \subseteq \darts{E}$ is a simple cycle in $G$ if and only if the corresponding dual darts comprise a simple cut in $G^{\ast}$.
\end{theorem}

\subsection{The left/right relation}
\label{sub_leftright}

\begin{assumption}
For the rest of this paper, let $G = (V, E)$ be a connected, planar graph, $s, t \in V$, with an embedding such that $t$ is adjacent to $f_{\infty}$. Furthermore, let  $\mathcal{P}$ be the set of all simple $s$-$t$-paths in $G$.
\end{assumption}

In order to define a partial order on $\mathcal{P}$, we consider the vector space that is spanned by the edges of the graph and the subspace spanned by all cycles. It is well-known that the (clockwise) boundaries of the non-infinite faces comprise a basis of this cycle space.

\begin{definition}
 For a path or cycle $P \subset \darts{E}$ we define the vector $\delta_{P} \in \mathbb{R}^{E}$ by $\delta_{P}(e) := 1$ if $\arc{e} \in P$,  $\delta_{P}(e) := -1$ if $\antiarc{e} \in P$ and $\delta_{P}(e) := 0$ otherwise. For a face $f \in V^{\ast}$ we define $\delta_{f}$ to be the vector corresponding to the set of darts in the clockwise boundary of $f$ (with antiparallel darts canceling out). For a vector $\delta \in \mathbb{R}^{E}$, we let $\delta(\arc{e}) := \delta(e)$, $\delta(\antiarc{e}) := -\delta(e)$ and define the set of darts induced by $\delta$ to be $\setof(\delta) := \{d \in \darts{E} : \delta(d) > 0\}$.
\end{definition}

\begin{theorem_definition}
\label{def_facepotentials} 
 The subspace $\cyclespace{G} := \vecspan \{\delta_{C} : C \text{ is a cycle in } G\}$ is called \emph{cycle space}. Its elements are called \emph{circulations}. The set $\{\delta_{f} : f \in V^{\ast} \setminus \{f_{\infty}\}\}$ is a basis of $\cyclespace{G}$. In particular, there is a unique linear mapping $\Phi: \cyclespace{G} \rightarrow \mathbb{R}^{V^{\ast}}$, such that $\Phi(\delta)(f_{\infty}) = 0$ and $\delta = \sum_{f \in V^{\ast}} \Phi(\delta)(f)\delta_{f}$ for all $\delta \in \cyclespace{G}$. The vector $\Phi(\delta)$ is called the \emph{face potential} of $\delta$.
\end{theorem_definition}

\begin{remark} 
 If $\phi = \Phi(\delta)$ for $\delta \in \cyclespace{G}$, then $\delta(d) = \phi(\rightf(d)) - \phi(\leftf(d))$ for all $d \in \darts{E}$.
\end{remark}

The left/right relation goes back to ideas of Khuller et al.~\cite{Kh93} and Weihe~\cite{We97}, and was specified for paths  Klein~\cite{Kl05}. It yields useful applications for shortest path and maximum flow computations in planar graphs (cf. \cite{Kl05} and \cite{BK06}, respectively) and is based on the face potentials introduced above. Intuitively, the definition states that $P \preceq Q$ if and only if the circulation consisting of $P$ and the reverse of $Q$ is clockwise (as positive face potentials correspond to clockwise circulations).

\begin{definition}[Left/right relation]
\label{def_leftrightrelation}
 Let $P, Q \in \mathcal{P}$. If $\Phi(\delta_{P} - \delta_{Q}) \geq 0$, we say that $P$ \emph{is left of} $Q$ and write $P \leftof Q$. If  $\Phi(\delta_{P} - \delta_{Q}) \leq 0$, we say that $P$ \emph{is right of} $Q$ and write $P \rightof Q$.
\end{definition}

Observe that, by flow conservation, $\delta_{P} - \delta_{Q}$ is a circulation and thus the above relation is well-defined. It is easy to verify that it indeed is a partial order on $\mathcal{P}$. 

\begin{remark}
\label{rem_facepotentials}
 When analyzing a circulation and its face potential, we can ignore edges that are not in the support of the circulation. The potential is equal on both sides of such an edge, so removing it from the graph -- which is contracting it in the dual graph -- yields a subgraph, on which we basically can apply the same face potential. In particular, $P \rightof Q$ in $G$ if and only if $P \rightof Q$ in every subgraph of $G$ containing $P$ and $Q$.
\end{remark}

 \section{Uppermost paths and the path lattice of an $s$-$t$-plane graph}
\label{sec_stplanar}

Intuitively speaking, the uppermost path of an $s$-$t$-plane graph, is the $s$-$t$-path forming its ``upper'' boundary in a drawing where $s$ is on the very left and $t$ is on the very right of the drawing. The idea goes back to Ford and Fulkerson, who used it to introduce the uppermost path algorithm for the maximum flow problem in $s$-$t$-planar graphs \cite{FF56}, which iteratively saturates the uppermost residual path. We will give a definition of the uppermost path in combinatorial terms and use it to characterize the left/right relation in $s$-$t$-plane graphs. This yields all the desired lattice properties of the partial order and thus shows that the uppermost path algorithm corresponds to the two-phase greedy algorithm, which also saturates the maximum (w.r.t.\ $\preceq$) ``residual'' element of the lattice in each iteration \cite{FP08}.

\begin{assumption}
 Throughout Section \ref{sec_stplanar}, we assume that the embedding of $G$ is $s$-$t$-planar.
\end{assumption}

\begin{theorem_definition}[Uppermost and lowermost path]
\label{thm_uppermost_path_existence}
 There is a unique path $U \in \mathcal{P}$ such that $\leftf(d) = f_\infty$ for all $d \in U$. It is called \emph{uppermost path} of $G$. There also is a unique path $L \in \mathcal{P}$ such that $\rightf(d) = f_\infty$ for all $d \in L$. This path is called \emph{lowermost path} of $G$.
\end{theorem_definition}

Clearly, if $U$ is the uppermost path of $G$, then $U$ is also the uppermost path of any subgraph of $G$ containing all edges of $U$. The following lemma goes back to Ford and Fulkerson \cite{FF56}, who proved it by geometric reasoning. We give an alternative proof using cycle/cut duality, and conclude two further auxiliary results that we will need later. 

\begin{lemma}
\label{lem_uppermostpath_cut}
 Let $C$ be a simple $s$-$t$-cut. There is exactly one dart $d_{l} \in C$ with $\leftf(d_{l}) = f_{\infty}$ and exactly one dart $d_{r} \in C$ with $\rightf(d_{r}) = f_\infty$.
\end{lemma}
\begin{proof}
 As the uppermost and lowermost path both cross $C$, two darts $d_{r}$ and $d_{l}$ with the desired properties exist. By cycle/cut duality, $C$ is a simple cycle in $G^{\ast}$. In particular, there is only one occurrence of $f_{\infty}$ as a head and one occurrence as a tail of a dart in this simple cycle. This implies the uniqueness of $d_{r}$ and $d_{l}$.\qed
\end{proof}

\begin{lemma}[Orientation lemma]
 \label{lem_no_reverse}\index{orientation lemma}
 Let $U$ be the uppermost path of $G$ and $L$ be the lowermost path of $G$. If $d \in U$ then $\rev{d} \notin L$.
\end{lemma}
\begin{proof}
 Assume by contradiction there is a dart $d$ in $U$ with $\rev{d} \in L$. Let $\Gamma^{+}(S)$ be a simple $s$-$t$-cut containing $d$. As $L$ starts at $s \in S$, it must cross the cut once before it uses $\rev{d}$ to go back from $V \setminus S$ to $S$, and cross it a second time before it ends at $t \in V \setminus S$, a contradiction to Lemma \ref{lem_uppermostpath_cut}.\qed
\end{proof}

\begin{lemma}[Bridge lemma]
 \label{lem_bridge_lemma}
 Let $d \in \darts{E}$. If $\leftf(d) = \rightf(d)$, then $d$ is either contained in all simple $s$-$t$-paths or in none.
\end{lemma}
\begin{proof}
 Since $d$ is a loop in the dual graph, it forms a one-dart simple cut in $G$. If a simple path crosses the cut, it cannot go back. So $P$ starts on one side of the cut and ends on the other, implying that the cut separates $s$ from $t$. Thus, every $s$-$t$-path has to use $d$.\qed
\end{proof}

This suffices to characterize the left/right relation in an $s$-$t$-plane graph in terms of the uppermost path property.

\begin{theorem}
 Let, $P, Q \in \mathcal{P}$. Then the following statements are equivalent.
 \begin{enumerate}
  \item $P$ is the uppermost path in $G[E(P \cup Q)]$.
  \item $Q$ is the lowermost path in $G[E(P \cup Q)]$.
  \item $P$ is left of $Q$.
 \end{enumerate}
\end{theorem}

\begin{proof}
  We can assume that $\phi := \Phi(\delta_{P} - \delta_{Q})$ is a potential in $G[E(P \cup Q)]$ by Remark \ref{rem_facepotentials}.
 \begin{prooflist}
  \item[$(1) \Leftrightarrow (2):$] Suppose $P$ is the uppermost path of $G[E(P \cup Q)]$. Let $L$ be the lowermost path of $G[E(P \cup Q)]$. Let $d \in L$. If $d$ belongs to an edge of $E(P)$, the orientation lemma ensures that $d \in P$. Consequently, $d \in P \cap L$, and thus $\leftf_{G[E(P \cup Q)]}(d) = f_{\infty} = \rightf_{G[E(P \cup Q)]}(d)$, implying $d \in Q$ by the bridge lemma. Hence $E(L) \subseteq E(Q)$, and as the two paths are simple, they are equal. The converse follows by symmetry.\footnote{This equivalency implies that no path $P$ can use a reverse dart of the uppermost path $U$, as $P$ is the lowermost path in $G[E(U \cup P)]$. From now on, we will implicitly use this stronger result when referring to the orientation lemma.\index{orientation lemma}}
  \item[$(2) \Rightarrow (3):$] Suppose $Q$ is the lowermost path of $G[E(P \cup Q)]$ and thus $P$ is its uppermost path. Let $f$ be a face of $G[E(P \cup Q])$ and let $d \in \darts{E}$ with $\rightf(d) = f$. If $\rev{d} \in P$ or $d \in Q$, then $f = \rightf(d) = f_{\infty}$, implying $\phi(f) = 0$. Otherwise, $d \in P$ or $\rev{d} \in Q$, implying $\leftf(d) = f_{\infty}$ and $\phi(f) = \phi(f_{\infty}) + \delta_{P}(d) - \delta_{Q}(d) \geq 0$. Thus, $\phi \geq 0$.
  \item[$(3) \Rightarrow (1):$] Suppose $\phi \geq 0$. Let $U$ be the uppermost path of $G[E(P \cup Q)]$ and $d \in U$. By the orientation lemma, $d \in P$ or $d \in Q$. But $d \in Q \setminus P$ is not possible, as $\delta_{P}(d) - \delta_{Q}(d) = \phi(\rightf(d)) - \phi(f_{\infty}) \geq 0$. So $d \in P$ for all $d \in U$, i.e., $U = P$.\qed
 \end{prooflist}
\end{proof}

Thus, the left/right order in $s$-$t$-plane graphs is in fact an uppermost/lowermost path order. Before we can show that this indeed yields a lattice, we need a final auxiliary result, which states that we can add a path to a subgraph without changing its uppermost path, as long as there already is a path above the path we add.

\begin{lemma}
 \label{lem_add_a_path}
 Let $\bar{E} \subseteq E$ be an edge set, such that $G[\bar{E}]$ is connected and let $P \in \mathcal{P}$ be an $s$-$t$-path with $E(P) \subseteq \bar{E}$. Let $Q \in \mathcal{P}$. If $P$ is left of $Q$, then the uppermost path of $G[\bar{E} \cup E(Q)]$ and the uppermost path of $G[\bar{E}]$ are equal.
\end{lemma}

\begin{proof}
 Let $U, U^{\prime}$ be the uppermost paths of $G[\bar{E}]$ and $G[\bar{E} \cup E(Q)]$, respectively and assume by contradiction that $U \neq U^{\prime}$. Then $U^{\prime}$ uses an edge in $E(Q) \setminus \bar{E}$, and by the orientation lemma, it even uses the same dart $d$ of the edge that is used by $Q$. For this dart, $\leftf_{G[\bar{E} \cup E(Q)]}(d) = f_\infty$, and hence $\leftf_{G[E(P \cup Q)]}(d) = f_{\infty}$, as $E(P \cup Q) \subseteq \bar{E} \cup E(Q)$. But as $Q$ is the lowermost path in $G[E(P \cup Q)]$, also $\rightf_{G[E(P \cup Q)]}(d) = f_{\infty}$. Thus, $d \in P$ by the bridge lemma, a contradiction, as $d$ was chosen as a dart not in $G[\bar{E}]$.\qed
\end{proof}

Finally, we can show the existence of a consecutive and submodular path lattice in an $s$-$t$-plane graph. We even get a nice characterization of meet and join of this lattice as the lowermost and uppermost path of $G[E(P \cup Q)]$.

\begin{theorem}
 $(\mathcal{P}, \rightof)$ is a consecutive and submodular lattice with $P \wedge Q$ being the lowermost path in $G[E(P \cup Q)]$ and $P \vee Q$ being the uppermost path in $G[E(P \cup Q)]$.
\end{theorem}
\begin{proof}
 We first show that meet and join can indeed be defined as claimed. Then we deduce consecutivity and submodularity. Let $P, Q, R \in \mathcal{P}$.
 \begin{prooflist}
  \item[Meet and join:] Let $U$ be the uppermost path in $G[E(P \cup Q)]$. Then $P, Q \rightof U$. Let $U^{\prime} \in \mathcal{P}$ with $P, Q \rightof U^{\prime}$. Then $U^{\prime}$ is the uppermost path of $G[E(P \cup U^{\prime}) \cup E(Q)]$ by Lemma \ref{lem_add_a_path}. As $U$ is contained in this graph, $U \rightof U^{\prime}$. Thus, $U$ is the least upper bound on $P$ and $Q$ with respect to $\rightof$. The meet follows by symmetry.
  \item[Consecutivity:] Suppose $P \rightof Q \rightof R$. Then $P$ is the lowermost path and $R$ is the uppermost path of $G^{\prime} := G[E(P \cup R) \cup E(Q)]$ by Lemma \ref{lem_add_a_path}. Thus, $\rightf_{G^{\prime}}(d) = f_{\infty} = \leftf_{G^{\prime}}(d)$ for all $d \in P \cap R$. By the bridge lemma, this implies $d \in Q$.
  \item[Submodularity:] As we have proven consecutivity, it suffices to show $P \wedge Q, P \vee Q \subseteq P \cup Q$. This immediately follows from the definition of $P \wedge Q$ and $P \vee Q$ as lowermost and uppermost path of $G[E(P \cup Q)]$ and the orientation lemma.\qed
 \end{prooflist}
\end{proof}

 \begin{figure}[t]
 \begin{center}
 \begin{tikzpicture}[scale=0.8]
	\node[labeledNode] (s) at (-5, 3) {s};
	\node[labeledNode] (t) at (0, 3) {t};
	\node[normalNode] (1) at (-2.5, 4) {}
		edge[<-, normalEdge] node[above, sloped] {$e_{1}$} (s)
		edge[->, normalEdge] node[above, sloped] {$e_{3}$} (t);
	\node[normalNode] (2) at (-2.5, 2) {}
		edge[<-, normalEdge] node[below, sloped] {$e_{2}$} (s)
		edge[->, normalEdge] node[below, sloped] {$e_{4}$} (t)
		edge[<-, normalEdge] node[left] {$e_{5}$} (1);
		
	\node[draw] (P1) at (3, 4) {$P_{1}$};
	\node[draw] (P2) at (2, 3) {$P_{2}$}
		edge[normalEdge] (P1);
	\node[draw] (P3) at (4, 3) {$P_{3}$}
		edge[normalEdge] (P1);
 \node[draw] (P4) at (3, 2) {$P_{4}$}
		edge[normalEdge] (P2)
		edge[normalEdge] (P3);

	\node at (1.5, 4) {$(\mathcal{P}, \rightof)$};
	
	\node at (-5.3, 4) {$G = (V, E)$};

\draw[gray] (-8.1, 1.3) rectangle (6.7,-0.2);
\draw[gray] (-4.4, 1.3) -- (-4.4, -0.2) (-0.7, 1.3) -- (-0.7, -0.2) (3, 1.3) -- (3, -0.2);
	
	\node[draw] (PP1) at (-7.5, 0.5) {$P_{1}$};
	\draw (PP1) + (0.7, 0) node[smallLabeledNode] (sP1) {};
	\draw (sP1) +(2, 0) node[smallLabeledNode] (tP1) {};
	\draw (sP1) +(1, 0.5) node[smallNode] (1P1) {}
		edge[<-, normalEdge] (sP1)
		edge[->, normalEdge] (tP1);
		
	\draw (PP1) +(3.7, 0) node[draw] (PP2) {$P_{2}$};
	\draw (PP2) + (0.7, 0) node[smallLabeledNode] (sP2) {};
	\draw (sP2) +(2, 0) node[smallLabeledNode] (tP2) {};
	\draw (sP2) +(1, 0.5) node[smallNode] (1P2) {}
		edge[<-, normalEdge] (sP2);
	\draw (sP2) +(1, -0.5) node[smallNode] (2P2) {}
		edge[<-, normalEdge] (1P2)
		edge[->, normalEdge] (tP2);
		
	\draw (PP2) +(3.7, 0) node[draw] (PP3) {$P_{3}$};
	\draw (PP3) + (0.7, 0) node[smallLabeledNode] (sP3) {};
	\draw (sP3) +(2, 0) node[smallLabeledNode] (tP3) {};
	\draw (sP3) +(1, 0.5) node[smallNode] (1P3) {}
		edge[->, normalEdge] (tP3);
	\draw (sP3) +(1, -0.5) node[smallNode] (2P3) {}
		edge[->, normalEdge] (1P3)
		edge[<-, normalEdge] (sP3);
		
	\draw (PP3) +(3.7, 0) node[draw] (PP4) {$P_{4}$};
	\draw (PP4) + (0.7, 0) node[smallLabeledNode] (sP4) {};
	\draw (sP4) +(2, 0) node[smallLabeledNode] (tP4) {};
	\draw (sP4) +(1, -0.5) node[smallNode] (2P4) {}
		edge[<-, normalEdge] (sP4)
		edge[->, normalEdge] (tP4);
\end{tikzpicture}
 \caption[test]{An $s$-$t$-plane graph and the lattice of its simple $s$-$t$-paths. The paths $P_{2}$ and $P_{3}$ are not comparable, but $P_{2} \wedge P_{3} = P_{4}$ and $P_{2} \vee P_{3} = P_{1}$.}
 \label{ex_pathlattice}
 \end{center}
 \end{figure}
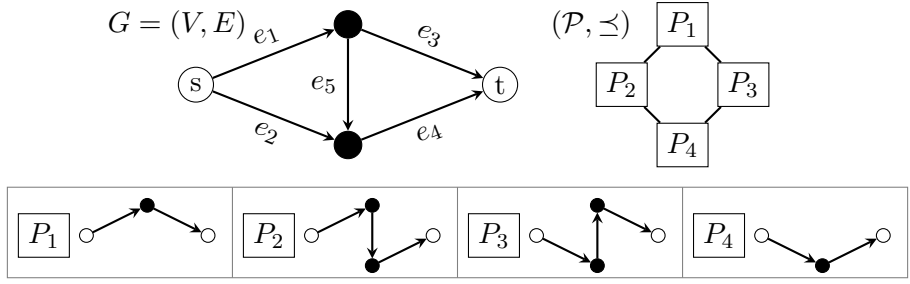

An example of a graph and its path lattice is depicted in Figure \ref{ex_pathlattice}.

Our result implies total dual integrality of the maximum flow problem in $s$-$t$-planar graphs, even when we introduce supermodular weights on the paths. Applying the two-phase greedy algorithm on the lattice yields an implementation of the uppermost path algorithm by Ford and Fulkerson that solves the maximum flow problem in $\bigo(|V| \log(|V|))$, again also for the case of supermodular and monotone increasing path weights (cf. \cite{Ma09} for more details).

\section{The path lattice of a general plane graph}
\label{sec_general}

We will now show that the left/right relation also defines a lattice in general plane graphs. However, the proof will require significantly more effort this time. In contrast to the $s$-$t$-planar case, meet and join of two paths $P, Q \in \mathcal{P}$ are not always the minimum (rightmost) and maximum (leftmost) path in $G[E(P \cup Q)]$ in the general case (cf. Figure \ref{ex_leftrightmeetjoin}
for an example). The intuitive idea for constructing the meet is the following: We subtract the ``positive part'' of the circulation $\delta_{P} - \delta_{Q}$ (i.e., those faces that prevent $P$ from being right of $Q$) from the path $P$. In this way we obtain a set of darts $D^{P \wedge Q} \subseteq P \cup Q$ that contains the meet $P \wedge Q$, as we shall see later. We formalize this idea in the following lemma.

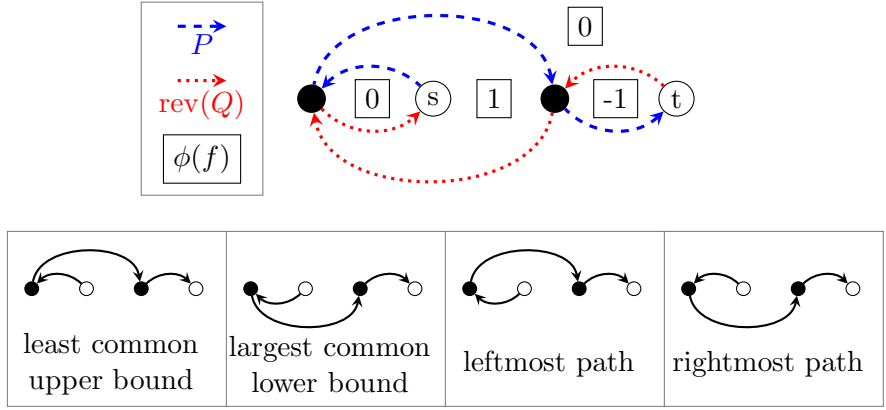
\begin{figure}[t]
\begin{center}
\begin{tikzpicture}[scale=0.8]
	\node[color=blue] at (-9.8, 0.9) {$P$};
	\draw [->, >=stealth, very thick, dashed, color=blue] (-10.2,1.2) to (-9.4,1.2);
	\node[color=red] at (-9.8, -0.1) {$\rev{Q}$};
	\draw [->, >=stealth, very thick, dotted, color=red] (-10.2,0.3) to (-9.4,0.3);
	\node[draw] (phi) at (-9.8,-1) {$\phi(f)$};
	\draw[gray] (-10.8, 1.6) rectangle ++(2, -3.2);
	
	\path (-6, 0) node[labeledNode] (s) {s}
		++(-2, 0) node[normalNode] (1) {}
			edge[<-, >=stealth, very thick, dashed, color=blue, bend left=45] (s)
			edge[->, >=stealth, very thick, dotted, color=red, bend right=45] (s)
		++(4,0) node[normalNode] (2) {}
			edge[<-, >=stealth, very thick, dashed, color=blue, bend right=80] (1)
			edge[->, >=stealth, very thick, dotted, color=red, bend left=80] (1)
		++(2,0) node[labeledNode] (t) {t}
			edge[<-, >=stealth, very thick, dashed, color=blue, bend left=45] (2)
			edge[->, >=stealth, very thick, dotted, color=red, bend right=45] (2)
		++(-1,0) node[dualNode] {-1} +(-0.5,1.2) node[dualNode] {0} 
		++(-2,0) node[dualNode] {1} ++(-2,0) node[dualNode] {0};
		
	\draw[gray] (-13, -2.2) rectangle (1.4, -5.1) (-9.4, -2.2) -- (-9.4, -5.1) (-5.8, -2.2) -- (-5.8, -5.1) (-2.2, -2.2) -- (-2.2, -5.1);	
	
	\path[scale=0.45] (-26, -7) node[smallLabeledNode] (p1s) {}
		++(-2, 0) node[smallNode] (p11) {}
			edge[<-, >=stealth, thick, bend left=45] (p1s)
		++(4,0) node[smallNode] (p12) {}
			edge[<-, >=stealth, thick, bend right=80] (p11)
		++(2,0) node[smallLabeledNode] (p1t) {}
			edge[<-, >=stealth, thick, bend right=45] (p12);
			
	\node[text width=3cm, text centered] at (-11.3, -4.4) {least common upper bound};
		
	\path[scale=0.45] (-18, -7) node[smallLabeledNode] (p2s) {}
		++(-2, 0) node[smallNode] (p21) {}
			edge[<-, >=stealth, thick, bend right=45] (p2s)
		++(4,0) node[smallNode] (p22) {}
			edge[<-, >=stealth, thick, bend left=80] (p21)
		++(2,0) node[smallLabeledNode] (p2t) {}
			edge[<-, >=stealth, thick, bend right=45] (p22);
			
	\node[text width=3cm, text centered] at (-7.7, -4.4) {largest common lower bound};
		
	\path[scale=0.45] (-10, -7) node[smallLabeledNode] (p3s) {}
		++(-2, 0) node[smallNode] (p31) {}
			edge[<-, >=stealth, thick, bend right=45] (p3s)
		++(4,0) node[smallNode] (p32) {}
			edge[<-, >=stealth, thick, bend right=80] (p31)
		++(2,0) node[smallLabeledNode] (p3t) {}
			edge[<-, >=stealth, thick, bend right=45] (p32);
		
	\node[text width=3cm, text centered] at (-4.1, -4.4) {leftmost path};
		
	\path[scale=0.45] (-2, -7) node[smallLabeledNode] (p4s) {}
		++(-2, 0) node[smallNode] (p41) {}
			edge[<-, >=stealth, thick, bend left=45] (p4s)
		++(4,0) node[smallNode] (p42) {}
			edge[<-, >=stealth, thick, bend left=80] (p41)
		++(2,0) node[smallLabeledNode] (p4t) {}
			edge[<-, >=stealth, thick, bend right=45] (p42);
			
	\node[text width=3cm, text centered] at (-0.5, -4.4) {rightmost path};
\end{tikzpicture}
\caption[test]{This example shows that the rightmost path of $G[E(P \cup Q)]$ is not necessarily the largest common lower bound of the two paths in non-$s$-$t$-planar embeddings. However, subtracting the ``positive'' part of the circulation from $P$ yields $P \wedge Q$.}
\label{ex_leftrightmeetjoin}
\end{center}
\end{figure}

\begin{lemma}
\label{lem_meetdefinition}
 Let $P, Q \in \mathcal{P}$ and $\phi := \Phi(\delta_{P} - \delta_{Q})$.
 \begin{itemize}
  \item Let $S^{+} := \{f \in V^{\ast} : \phi(f) > 0\}$ and $\delta^{P \wedge Q} := \delta_{P} - \sum_{f \in S^{+}} \phi(f) \delta_{f}$. Then $\delta^{P \wedge Q} \in \{-1, 0, 1\}^{E}$ and $D^{P \wedge Q} := \setof(\delta^{P \wedge Q}) \subseteq P \cup Q$.
  \item Let $S^{-} := \{f \in V^{\ast} : \phi(f) < 0\}$ and $\delta^{P \vee Q} := \delta_{P} - \sum_{f \in S^{-}} \phi(f) \delta_{f}$. Then $\delta^{P \vee Q} \in \{-1, 0, 1\}^{E}$ and $D^{P \vee Q} := \setof(\delta^{P \vee Q}) \subseteq P \cup Q$. 
 \end{itemize}
\end{lemma}

It is straightforward to check that if $P \preceq Q$, then $P = D^{P \wedge Q}$ and $Q = D^{P \vee Q}$. Unfortunately, $D^{P \wedge Q}$ and $D^{P \vee Q}$ are not $s$-$t$-paths in general. However, it can be shown that $D^{P \wedge Q}$ and $D^{P \vee Q}$ each consist of a unique simple $s$-$t$-path and some cycles and that these paths are meet and join of $P$ and $Q$, respectively. From now on, we will focus on the set $D^{P \wedge Q}$, but analogous versions of all statements can be shown for $D^{P \vee Q}$. The proof of Lemma \ref{lem_meetdefinition} can be obtained by a simple case distinction, and, as a by-product, leads to the following additional result.

\begin{lemma}
\label{lem_no_crossing}
 If $R \subseteq D^{P \wedge Q}$ is a simple path and $C \subseteq D^{P \wedge Q}$ is a simple cycle such that $R \cap C = \emptyset$, then $R$ does not cross $C$, i.e., all darts of $R$ are either in the interior of $C$ or none of them is.
\end{lemma}

The following lemma is the key insight on our way to proving the desired result. Its proof, however, is rather lengthy and involves many minor details. The key idea is that such a cycle must consist of edges of $P$ and $Q$ and that these paths can only enter or leave the cycle from or to the left, i.e., from or towards its interior, thus being ``trapped'' inside the cycle. This provides us with a contradiction.

\begin{lemma}\label{lem_nocounterclockwise}
 There are no counterclockwise simple cycles in $D^{P \wedge Q}$.
\end{lemma}

\begin{theorem}
\label{thm_leftmostpathlattice}
 $(\mathcal{P}, \preceq)$ is a submodular lattice with $P \wedge Q$ being the unique simple $s$-$t$-path contained in $D^{P \wedge Q}$ and $P \vee Q$ being the unique simple $s$-$t$-path contained in $D^{P \vee Q}$.
\end{theorem}
\begin{proof} 
 As $\delta^{P \wedge Q} \in \{-1, 0, 1\}^{E}$ is the sum of $\delta_{P}$ and some circulations, $\delta^{P \wedge Q}$ is a unit $s$-$t$-flow of value $1$. Thus, there is a flow decomposition $\delta_{R} + \sum_{i=1}^{k} \delta_{C_{i}} = \delta^{P \wedge Q}$ for a simple $s$-$t$-path $R \subseteq D^{P \wedge Q}$ and some -- by Lemma \ref{lem_nocounterclockwise} clockwise -- simple cycles $C_{1}, \ldots, C_{k} \subseteq D^{P \wedge Q}$ with $E(R), E(C_{1}), \ldots, E(C_{k})$ pairwise disjoint.
 
 We now show $R \preceq P$ and $R \preceq Q$. This is directly implied by the two inequalities stated below, which follow from the linearity of $\Phi$, the equality $\phi = \sum_{f \in S^{+}} \phi(f) \Phi(\delta_{f}) + \sum_{f \in S^{-}} \phi(f) \Phi(\delta_{f})$, and the fact that $\Phi(\delta_{C}) \geq 0$ and $\Phi(\delta_{f}) \geq 0$ for any clockwise simple cycle $C$ and any face $f \in V^{\ast}$.
 
 \begin{eqnarray*}
  \Phi(\delta_{P} - \delta_{R}) & = & \Phi\left(\delta_{P} - \left(\delta^{P \wedge Q} - \sum_{i=1}^{k}\delta_{C_{i}}\right)\right) = \sum_{f \in S^{+}} \phi(f) \Phi(\delta_{f}) + \sum_{i = 1}^{k} \Phi(\delta_{C_{i}}) ~~ \geq ~~ 0\\
  \Phi(\delta_{Q} - \delta_{R}) & = & \Phi\left(\delta_{Q} - \left(\delta^{P \wedge Q} - \sum_{i=1}^{k}\delta_{C_{i}}\right)\right) = -\sum_{f \in S^{-}} \phi(f) \Phi(\delta_{f}) + \sum_{i = 1}^{k} \Phi(\delta_{C_{i}}) ~~ \geq ~~ 0
 \end{eqnarray*}
 
 Now let $S \in \mathcal{P}$ be a path with $S \preceq P$ and $S \preceq Q$. We show $S \preceq R$. First, we consider only faces incident to $R$. So let $\bar{f} \in \{\leftf(d) : d \in R\} {\cup} \{\rightf(d) : d \in R\}$. By Lemma \ref{lem_no_crossing}, $R$ cannot cross any of the cycles $C_{i}$, and, as $t$ is on the exterior of any such cycle, $R$ cannot have any darts in the interior of a cycle. Thus, $\bar{f}$ is not in the interior of any of the cycles $C_{i}$ and $\Phi(\delta_{C_{i}})(\bar{f}) = 0$. This yields
 
 \begin{eqnarray*}
  \Phi(\delta_{R} - \delta_{S})\left(\bar{f}\right) = \Phi\left(\delta_{P} - \sum_{f \in S^{+}} \phi(f) \delta_{f} - \delta_{S} \right)\left(\bar{f}\right) = \begin{cases} \Phi(\delta_{P} - \delta_{S})\left(\bar{f}\right) & \text{ if } \bar{f} \notin S^{+} \\
  \Phi(\delta_{Q} - \delta_{S})\left(\bar{f}\right) & \text{ if } \bar{f} \in S^{+}
  \end{cases} ~~ \geq 0.
 \end{eqnarray*}
 
 Now let $\hat{f} \in V^{\ast}$ be any face. As $S$ does not contain a cycle in the primal graph, it does not contain a cut in the dual graph. Thus, there is a path in the dual that leads from $\hat{f}$ to some face $\bar{f}$ incident to $R$ and does not intersect $S$ or $R$, i.e., the potential does not change along the path. Thus $\Phi(\delta_{R} - \delta_{S})(\hat{f}) = \Phi(\delta_{R} - \delta_{S})(\bar{f}) \geq 0$. Consequently, $S \preceq R$.
 
 We thus have shown that $R$ is the meet of $P$ and $Q$ ($R$ is unique by anti-symmetry). Likewise, we can show that $D^{P \vee Q}$ contains a unique $s$-$t$-path that is the least common upper bound of $P$ and $Q$. Thus, $(\mathcal{P}, \preceq)$ is a lattice with meet and join as described above.\qed
\end{proof}

\section{A characterization of $s$-$t$-planar graphs}
\label{sec_characterization}

It is easy to observe that the path lattice induced by the left/right relation in general planar (but non-$s$-$t$-planar) graphs is not necessarily consecutive (cf. the paths $P_{1}, P_{2}, P_{4}$ in Figure \ref{fig_counterexample}). Of course one might ask whether this property can be achieved by a different partial order on the paths. Indeed, one can show that no partial order in any planar but not $s$-$t$-planar graph can induce a lattice that is submodular and consecutive at the same time.

The key idea to proving this negative result is to show it for two graphs that comprise the $s$-$t$-planar equivalent to the famous Kuratowski graphs $K_{3,3}$ and $K_{5}$. So let $K^{s-t}_{3,3}$ and $K^{s-t}_{5}$ be the graphs that arise from the respective Kuratowski graphs by deleting the edge connecting $s$ and $t$ ($K^{s-t}_{3,3}$ is depicted in Figure \ref{fig_counterexample}).

 \begin{figure}[t]
  \begin{center}
   \begin{tikzpicture}[scale=0.8]
	\node[labeledNode] (s) at (-2, 0) {$s$};
	\node[normalNode] (1) at (0, -1.5) {}
		edge[<-, normalEdge] node[above, sloped] {$e_{2}$} (s);
	\node[normalNode] (2) at (0, 0) {}
		edge[<-, normalEdge] node[above, sloped] {$e_{1}$} (s);
	\node[normalNode] (3) at (2, -1.5) {}
		edge[<-, normalEdge] node[above, sloped] {$e_{5}$} (1)
		edge[<-, normalEdge] node[above, sloped] {$e_{4}$} (2);
	\node[normalNode] (4) at (0, 1.5) {}
		edge[<-, normalEdge] node[right] {$e_{3}$} (2);
	\node[labeledNode] (t) at (4,-1.5) {$t$}
		edge[<-, normalEdge] node[above, sloped] {$e_{7}$} (4)
		edge[<-, normalEdge] node[above, sloped] {$e_{8}$} (3);
	\draw[->, normalEdge] (0, -1.5) arc (300:64:1.8);
	\node at (-2.9, 0) {$e_{6}$};
		
	\path (6.5,0.1) node[draw] (P2) {$P_{1}$};
	\path (P2)++(-1.1,1.2) node[smallLabeledNode] (sP2) {} 
		++(0.8,-0.6) node[smallNode] (vP23) {}
		edge[<-, normalEdge] (sP2)
		++(0,1.2) node[smallNode] (vP24) {}
		++(1.6,-1.2) node[smallLabeledNode] (tP2) {}
		edge[<-, normalEdge] (vP24);
		\draw[->, normalEdge] (vP23) arc (300:64:0.72);
		
	\path (P2)+(3.5,0) node[draw] (P3) {$P_{2}$};
	\path (P3)++(-1.1,1.2) node[smallLabeledNode] (sP3) {} 
		+(0.8,0) node[smallNode] (vP31) {}
		edge[<-, normalEdge] (sP3)
		+(0.8,0.6) node[smallNode] (vP32) {}
		edge[<-, normalEdge] (vP31)
		+(2.4,-0.6) node[smallLabeledNode] (tP3) {}
		edge[<-, normalEdge] (vP32);
		
	\path (P2)+(0,-2.1) node[draw] (P6) {$P_{3}$};
	\path (P6)++(-1.1,1.2) node[smallLabeledNode] (sP6) {} 
		+(0.8,0) node[smallNode] (vP61) {}
		edge[<-, normalEdge] (sP6)
		+(1.6,-0.6) node[smallNode] (vP62) {}
		edge[<-, normalEdge] (vP61)
		+(2.4,-0.6) node[smallLabeledNode] (tP6) {}
		edge[<-, normalEdge] (vP62);
		
	\path (P6)+(3.5,0) node[draw] (P7) {$P_{4}$};
	\path (P7)++(-1.1,1.2) node[smallLabeledNode] (sP7) {} 
		+(0.8,-0.6) node[smallNode] (vP71) {}
		edge[<-, normalEdge] (sP7)
		+(1.6,-0.6) node[smallNode] (vP72) {}
		edge[<-, normalEdge] (vP71)
		+(2.4,-0.6) node[smallLabeledNode] (tP7) {}
		edge[<-, normalEdge] (vP72);
	
\end{tikzpicture}
   \caption[test]{The graph $K^{s-t}_{3,3}$ and four of its $s$-$t$-paths, which suffice for showing that there is no partial order that induces a submodular and consecutive lattice on the set of paths.}
   \label{fig_counterexample}
  \end{center}
 \end{figure}
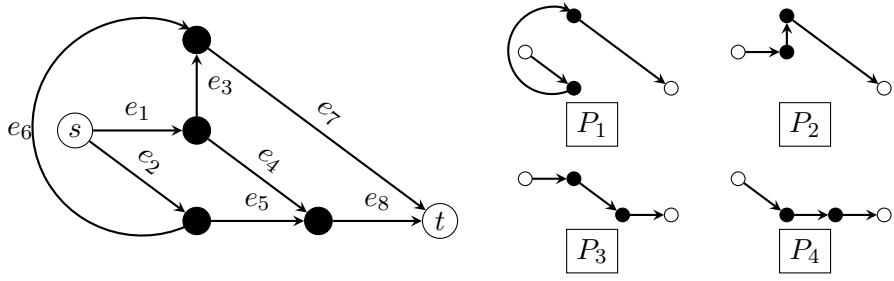

\begin{lemma}
\label{lem_comparability}
 Let $P, Q \in \mathcal{P}$ such that the subgraph $G[E(P \cup Q)]$ contains only the two paths $P$ and $Q$. If $\preceq$ is a partial order that induces a submodular lattice, then $P \preceq Q$ or $P \succeq Q$.
\end{lemma}
 \begin{proof}
  By submodularity, $P \vee Q \subseteq P \cup Q$ and thus $P \vee Q$ is a path in $G[E(P \cup Q)]$. This implies $P \vee Q = P$ or $P \vee Q = Q$.\qed
 \end{proof}

\begin{lemma}
\label{lem_counterexample}
  The set of $s$-$t$-paths in $K^{s-t}_{3,3}$ or $K^{s-t}_{5}$ (or a subdivision of these graphs) cannot be equipped with a partial order $\preceq$, such that $(\mathcal{P}, \preceq)$ is a consecutive and submodular lattice.
\end{lemma}
 \begin{proof}
  Assume by contradiction $\preceq$ is defined such that $(\mathcal{P}, \preceq)$ is a consecutive and submodular lattice with meet $\wedge$ and join $\vee$ on the set of $s$-$t$-paths in $K^{s-t}_{3,3}$. Consider the four paths $P_{1}, P_{2}, P_{3}, P_{4}$ depicted in Figure \ref{fig_counterexample}. It can easily be checked that $P_{i}$ and $P_{j}$ are the only two $s$-$t$-paths in $G[E(P_{i} \cup P_{j})]$ for $i \neq j$. Thus, by Lemma \ref{lem_comparability}, the paths form a chain w.r.t.\ $\preceq$. Since $P_{1}$ and $P_{2}$ are the only two of the paths sharing $\arc{e_{7}}$ as common dart, and $P_{2}$ and $P_{3}$ are the only two of the paths sharing $\arc{e_{1}}$, and $P_{3}$ and $P_{4}$ are the only two of the paths sharing $\arc{e_{8}}$, consecutivity demands that either $P_{1}\ \succ\ P_{2}\ \succ\ P_{3}\ \succ\ P_{4}$ or $P_{1}\ \prec\ P_{2}\ \prec\ P_{3}\ \prec\ P_{4}$. In both cases $\arc{e_{2}} \in P_{1} \cap P_{4} \setminus P_{2}$ yields a contradiction to consecutivity.
  Note that all arguments used in this proof are invariant under the operation of subdividing edges. The result for $K^{s-t}_{5}$ can be derived by a similar line of argumentation.\qed
 \end{proof}

\begin{theorem}
 A graph is $s$-$t$-planar if and only if it is planar and there is a partial order on the set of its $s$-$t$-paths that induces a consecutive and submodular lattice.
\end{theorem}
 \begin{proof}
  Necessity follows from Theorem \ref{thm_leftmostpathlattice}. For sufficiency, let $G$ be a graph that is planar but not $s$-$t$-planar. Let $G^{\prime}$ be the graph that is obtained from $G$ by adding an edge $e$ from $s$ to $t$. As $G$ is not $s$-$t$-planar, $G^{\prime}$ is not planar and thus, by Kuratowski's theorem, it contains a subdivision $K^{\prime}$ of $K_{3,3}$ or $K_{5}$. As $G$ is planar, one of the subdivided edges in $K^{\prime}$ must contain $e$. Let $s^{\prime}$ and $t^{\prime}$ be the endpoints that are connected by this subdivided edge. Clearly, $G$ contains a subdivision $K$ of $K^{s^{\prime}-t^{\prime}}_{3,3}$ or $K^{s^{\prime}-t^{\prime}}_{5}$, respectively. By Lemma \ref{lem_counterexample}, there is a set of $s^{\prime}$-$t^{\prime}$-paths in $K$ (and thus in $G$) that cannot be equipped with any partial order that induces a consecutive and submodular lattice. These paths can all be extended to $s$-$t$-paths in $G$ by using the $s$-$s^{\prime}$-path and the $t^{\prime}$-$t$-path contained in the subdivided edge connecting $s^{\prime}$ and $t^{\prime}$ in $K$. Thus $G$ contains a set of $s$-$t$-paths that cannot be equipped with a partial order of a consecutive and submodular lattice.\qed
 \end{proof}
 
\section{Conclusion and outlook}

We have established a connection between optimization on lattice polyhedra and planar network flow theory by showing that the structure exploited by two important flow algorithms corresponds to a submodular lattice. For the $s$-$t$-planar case, this implies that the uppermost path algorithm of Ford and Fulkerson is a special case of a more general two-phase greedy algorithm, which allows for certain types of weights on the lattice sets. Thus, a closer study of the weighted maximum flow problem is of obvious interest. First results in this direction can be found in \cite{Ma09}. Future research could also deal with the question in how far the structural result presented here lead to new insights for existing or new planar graph algorithms. Finally, it might be interesting to investigate whether the path lattice is distributive and our results extend a line of distributive lattices in planar graph structures connected to the left/right relation \cite{FK09}.

\newpage

\appendix

\section*{Appendix}
\setcounter{section}{1}

\subsection*{Appendix to Section \ref{sec_stplanar}: Existence and uniqueness of the uppermost path}

\textbf{Theorem and Definition \ref{thm_uppermost_path_existence}} (Uppermost and lowermost path)\textbf{.} \textit{There is a unique path $U \in \mathcal{P}$ such that $\leftf(d) = f_\infty$ for all $d \in U$. It is called \emph{uppermost path} of $G$. There also is a unique path $L \in \mathcal{P}$ such that $\rightf(d) = f_\infty$ for all $d \in L$. This path is called \emph{lowermost path} of $G$.
}
\begin{proof}
 For existence it suffices to traverse the closed walk that comprises the boundary of the infinite face starting at $s$ in such a direction that $f_{\infty}$ is always on the left. By $s$-$t$-planarity, $t$ occurs on that walk, and by smoothing out cycles we get a simple $s$-$t$-path.
 
 Now suppose by contradiction there are two distinct paths $U_{1}, U_{2} \in \mathcal{P}$ with the uppermost path property. Then $\delta_{U_{1}} - \delta_{U_{2}} \neq 0$ is a circulation and thus contains a simple cycle $C$, which contains at least one dart $d_{1} \in U_{1}$ and one dart $d_{2} \in \rev{U_{2}}$, as neither of the paths contains a cycle on its own. Thus $\leftf(d_{1}) = f_{\infty}$ and $\rightf(d_{2}) = \leftf(\rev{d_{2}}) = f_{\infty}$, a contradiction. By the same arguments the lowermost path exists and is unique.
\end{proof}

\subsection*{Appendix to Section \ref{sec_general}: The proofs of the Lemmas \ref{lem_meetdefinition}, \ref{lem_no_crossing} and \ref{lem_nocounterclockwise}, and submodularity of the path lattice}

In Section \ref{sec_general}, we stated three lemmas without giving thorough proofs. We also omitted the submodularity of the path lattice in the proof of Theorem \ref{thm_leftmostpathlattice}. We make up for this in this part of the appendix. As already stated in the main article, the proofs will be of a certain length and not always of the intuitive kind.

First, we restate and prove Lemma \ref{lem_meetdefinition}. The proof consists of a simple case distinction, as promised in the main article.
\newline\newline
\textbf{Lemma \ref{lem_meetdefinition}.} 
\textit{
	Let $P, Q \in \mathcal{P}$ and $\phi := \Phi(\delta_{P} - \delta_{Q})$.
	\begin{itemize}
		\item Let $S^{+} := \{f \in V^{\ast} : \phi(f) > 0\}$ and $\delta^{P \wedge Q} := \delta_{P} - \sum_{f \in S^{+}} \phi(f) \delta_{f}$. Then $\delta^{P \wedge Q} \in \{-1, 0, 1\}^{E}$ and $D^{P \wedge Q} := \setof(\delta^{P \wedge Q}) \subseteq P \cup Q$.
		\item Let $S^{-} := \{f \in V^{\ast} : \phi(f) < 0\}$ and $\delta^{\vee} := \delta_{P} - \sum_{f \in S^{-}} \phi(f) \delta_{f}$. Then $\delta^{\vee} \in \{-1, 0, 1\}^{E}$ and $D^{\vee} := \setof(\delta^{\vee}) \subseteq P \cup Q$.
	\end{itemize}
}
\begin{proof}
	We only show the first statement, the second one follows analogously.
	
	Let $d \in \darts{E}$ and $r := \rightf(d), l := \leftf(d)$. We show that $\delta^{P \wedge Q}(d) \in \{-1, 0, 1\}$ and that $\delta^{P \wedge Q}(d) = 1$ implies $d \in P \cup Q$. Note that \[\sum_{f \in S^{+}} \phi(f) \delta_{f}(d) = \sum_{f \in V^{\ast}} \mathbbm{1}_{S^{+}}(f) \phi(f) \delta_{f}(d) = \mathbbm{1}_{S^{+}}(r) \phi(r) - \mathbbm{1}_{S^{+}}(l) \phi(l).\]
	Consider the four possible cases:
	
	\begin{enumerate}
		\item $l, r \in S^{+}$:	We have $\delta^{P \wedge Q}(d) = \delta_{P}(d) - (\phi(r) - \phi(l)) = \delta_{P}(d) - (\delta_{P} - \delta_{Q})(d) = \delta_{Q}(d) \in \{-1, 0, 1\}$.
		
		\item $l \in S^{+}, r \notin S^{+}$: In this case $\delta^{P \wedge Q}(d) = \delta_{P}(d) + \phi(l)$. On the one hand, $\delta_{P}(d) + \phi(l) > -1$ as $\delta_{P}(d) \geq -1$ and $\phi(l) > 0$. On the other hand, $\delta_{P}(d) + \phi(l) \leq \delta_{P}(d) - (\phi(r) - \phi(l)) = \delta_{P}(d) - (\delta_{P} - \delta_{Q})(d) = \delta_{Q}(d) \leq 1$, as $\phi(r) \leq 0$. So $\delta^{P \wedge Q}(d) \in \{0, 1\}$, and, in particular, $\delta^{P \wedge Q}(d) = 0$ if $d \notin Q$.
		
		\item $r \in S^{+}, l \notin S^{+}$: This is equivalent to case (2) holding for $\rev{d}$. In particular $\delta^{P \wedge Q}(d) = -\delta^{P \wedge Q}(\rev{d}) \in \{-1, 0\}$.
		\item $l, r \notin S^{+}$: We have $\sum_{f \in S^{+}} \phi(f) \delta_{f}(d) = 0$, and hence $\delta^{P \wedge Q}(d) = \delta_{P}(d) \in \{-1, 0, 1\}$.
	\end{enumerate}
	Note that all cases with $\delta^{P \wedge Q}(d) = 1$ require $d \in P$ or $d \in Q$, so $D^{P \wedge Q} \subseteq P \cup Q$.
\end{proof}

Before we can proceed to the remaining open proofs of this section, we need to introduce some additional notations.

\subsubsection*{Simplifying assumptions and notations}

In some of the proofs we will also construct cycles and paths from parts of other cycles and paths. Thus, we need the following notation.

\begin{notation}
 If $P = d_{1}, \ldots, d_{k}$ is an $x$-$y$-walk and $Q = d^{\prime}_{1}, \ldots, d^{\prime}_{k^{\prime}}$ is a $y$-$z$-walk, we denote the $x$-$z$ walk $d_{1}, \ldots, d_{k}, d^{\prime}_{1}, \ldots, d^{\prime}_{k^{\prime}}$ by $P \circ Q$. We denote the $y$-$x$-walk $\rev{d_{k}}, \ldots, \rev{d_{1}}$ by $\rev{P}$. If $P$ is a simple path or cycle and $\tail(d_{i}) = v$, $\head(d_{j}) = w$ for $i < j$, we denote the $v$-$w$-path $d_{i}, \ldots, d_{j}$ by $P[v, w]$.
\end{notation}

Furthermore, we will let $\pi$ be the \emph{rotation system} of the embedded graph, i.e., $\pi$ is a permutation on the set of darts that gives for each vertex the counterclockwise order in which the darts leave that vertex in the embedding. Intuitively, if some observer is standing on a vertex, looking in direction of the dart $d$ leaving that verex, then $\pi(d)$ is the next dart he is seeing when turning his sight in counterclockwise direction. This helps us to formalize the intuitive notions of darts touching a path or cycle from the left.

\begin{notation}
	Let $d \in \darts{E}$ be a dart, $P$ be a simple path or cycle and $v \in V$ be a vertex on $P$ such that $\head(d_{1}) = v = \tail(d_{2})$ for $d_{1}, d_{2} \in P$. We say that $d$ \emph{leaves $P$ to the left} at $v$ if
	\[\min \{k \in \mathbb{N}: \pi^{k}(d_{2}) = d\} < \min \{k \in \mathbb{N}: \pi^{k}(d_{2}) = \rev{d_{1}}\}.\]
	We say $d$ \emph{enters $P$ to the left} at $v$ if $\rev{d}$ leaves $P$ to the left at $v$. 
\end{notation}

The auxilliary results in the rest of this section will exclusively deal with the structure of $D^{P \wedge Q}$. Their proofs will only use darts of $P$ and $Q$ and the face potential $\phi$. For this reason it is convenient to assume that $G$ only contains the edges of $P$ and $Q$, as we can identify $\phi$ with a face potential in this graph by Remark \ref{rem_facepotentials}. All results can be obtainend analogously for $D^{P \vee Q}$ with exchanged occurrences of ``left'' and ``right'', ``clockwise'' and ``counterclockwise''.

\begin{assumption}
	We assume that $G = G[E(P \cup Q)]$.
\end{assumption}

In order to keep the proofs as understandable as possible, we introduce some further helpful naming conventions. It is important to keep in mind that some darts of $P \cup Q$ are in $D^{P \wedge Q}$, while others may not.

\begin{notation}
	For $d \in \darts{E}$ we define $\phi_{l}(d) := \phi(\leftf(d))$ and $\phi_{r}(d) := \phi(\rightf(d))$. We call a dart $d \in D^{P \wedge Q}$ \emph{solid} and a dart $d \in (P \cup Q) \setminus D^{P \wedge Q}$ \emph{hidden}. A dart in $P \setminus Q$ is called \emph{$P$-dart}, while a dart in $Q \setminus P$ is called \emph{$Q$-dart}. 
\end{notation}
	
Note that both solid and hidden darts can influence the face potential. Darts in $P \cap Q$ which are neither $P$- nor $Q$-darts have no influence on the potential, so they are rather uninteresting -- we will even eleminate them in a later assumption. Note that a dart can be a $P$-dart while its reverse can be a $Q$-dart, but only one of the two darts can be solid. Crossing a $P$-dart from right to left means that the potential decreases by $1$, and crossing a $Q$-dart from right to left means that the potential increases by 1 (so in the case of crossing an edge with anti-parallel $P$- and $Q$-darts the potential changes by $2$, accordingly). We make a first use of our definition by formulating an insight  obtained from the proof of Lemma \ref{lem_meetdefinition}.

\begin{corollary}
\label{cor_leftpotential}
	If $d$ is a solid $P$-dart, then $\phi_{l}(d) < 0$. If $d$ is a solid $Q$-dart, then $\phi_{l}(d) > 0$.
\end{corollary}
\begin{proof}
	If $d$ is a solid $P$-dart, i.e., $d \in D^{P \wedge Q} \setminus Q$, then case (4) must apply to $d$ in the proof of Lemma \ref{lem_meetdefinition}. Thus $\rightf(d) \notin S^{+}$ and $\phi_{l}(d) = \phi_{r}(d) - 1 < 0$ in this case. If $d$ is solid $Q$-dart, i.e., $d \in D^{P \wedge Q} \setminus P$, then case (1) or (2) must apply to $d$ in the proof of Lemma \ref{lem_meetdefinition}. Thus $\leftf(d) \in S^{+}$ and $\phi_{l}(d) > 0$ in this case.
\end{proof}

A more intuitive reason for the above result is, that a $Q$-dart can only become solid if it is part of a clockwise (positive potential) cycle in $\delta_{P} - \delta_{Q}$, which is then subtracted from $\delta_{P}$ to obtain $\delta^{P \wedge Q}$. As a direct implication of this corollary, we get Lemma \ref{lem_no_crossing}.
\newline\newline
\textbf{Lemma \ref{lem_no_crossing}.} \textit{If $R \subseteq D^{P \wedge Q}$ is a simple path and $C \subseteq D^{P \wedge Q}$ is a simple cycle such that $R \cap C = \emptyset$, then $R$ does not cross $C$, i.e., all darts of $R$ are either in the interior of $C$ or none of them is.}
\begin{proof}
 By contradiction assume there is a vertex $v \in V$ where $R$ and $C$ cross. Then $v$ is incident to four solid darts, in particular there is a solid $P$-dart $d_{P}$ and a solid $Q$-dart $d_{Q}$ leaving from $v$, one of them comming from the path, the other from the cycle. As the path and the cycle cross in $v$, we must have either $\pi(d_{P}) = d_{Q}$ or $\pi(d_{Q}) = d_{P}$. In the first case $\phi_{l}(d_{P}) = \phi_{r}(d_{Q}) = \phi_{l}(d_{Q}) - 1 \geq 0$ by Corollary \ref{cor_leftpotential}, which is a contradition to the very same corollary. The second case follows analogously.
\end{proof}

Indeed, it can even be shown that the order of incoming and outgoing darts at vertices incident to four solid darts is fixed (cf. Figure \ref{fig_ccwcycle_proof} (a) for an illustration).

\subsubsection*{``Change of tracks'' and the non-existence of solid counterclockwise cycles}

We now want to investigate those vertices at which solid $P$-darts and $Q$-darts meet more closely. We call the result the ``change of tracks'' lemma, as this is the exact situation on which we will apply it -- a change of tracks on a path or cycle, when a $P$-dart is followed by a $Q$-dart or vice versa. It basically states the following observation: In case there is at least one solid $P$-dart and one solid $Q$-dart at a vertex, we can deduce that the left potential of the $P$-dart is $-1$ and the left potential of the $Q$-dart is $1$. This change of potential implies that there are two additional (possibly hidden) darts of $P \cup Q$ at the particular vertex, one entering and one leaving the path to the left (cf. Figure \ref{fig_ccwcycle_proof} (b) for an illustration of this constellation).

\begin{lemma}[Change of tracks lemma]
\label{lem_changeoftracks}\index{change of tracks lemma}
	Let $d_{P}$ be a solid $P$-dart and $d_{Q}$ be a solid $Q$-dart. If  $\head(d_{P}) = \tail(d_{Q})$, then $\{d_{Q}, \pi(d_{Q}), \ldots, \rev{d_{P}}\}$ contains a $P$-dart and the reverse of a $Q$-dart. If $\head(d_{Q}) = \tail(d_{P})$, then $\{d_{P}, \pi(d_{P}), \ldots, \rev{d_{Q}}\}$ contains a $Q$-dart and the reverse of a $P$-dart. In both cases, $\phi_{l}(d_{P}) = -1$ and $\phi_{l}(d_{Q}) = 1$.
\begin{proof}
	We only show the statement for the case $v := \head(d_{P}) = \tail(d_{Q})$. Since $\phi_{l}(d_{P}) < 0$ and $\phi_{l}(d_{Q}) > 0$, the potential decreases by at least $2$ when traversing the faces $\leftf(d_{Q})$, $\leftf(\pi(d_{Q}))$, $\ldots$, $\rightf(\rev{d_{P}})$. However, the sequence $d_{Q}, \pi(d_{Q}), \ldots, \rev{d_{P}}$ contains at most four darts, as there can be only each one incoming and each one outgoing dart from $P$ and $Q$ at $v$, respectively. As the incoming dart from $P$ is $d_{P}$ and the outgoing dart from $Q$ is $d_{Q}$, there is only an outgoing dart $d_{P}^{\prime} \in P$ and an incoming dart $d_{Q}^{\prime} \in Q$ that can decrease the potential each by $1$ on the counterclockwise traversal of the faces from $\leftf(d_{Q})$ to $\leftf(d_{P})$. So $d_{P}^{\prime}$ and $\rev{d_{Q}^{\prime}}$ must both be contained in $\{d_{Q}, \pi(d_{Q}), \ldots, \rev{d_{P}}\}$ (note that $d_{P}^{\prime} = \rev{d_{Q}^{\prime}}$ is perfectly possible). As they cannot change the potential by more than $2$, we  have $\phi_{l}(d_{P}) = -1$ and $\phi_{l}(d_{Q}) = 1$ as claimed.
\end{proof}
\end{lemma}

\begin{figure}[t]
 \begin{center}

\begin{tikzpicture}
	
	\node[smallNode] (vchange) at (-4, -0.5) {};
	\draw (vchange) +(-2.5, 0.75) node {(b)};
	\draw[normalEdge, blue, very thick, <-] (vchange) -- +(-1.5, 0) node[midway, below] {$d_{P}$};
	\draw[normalEdge, red, very thick, ->] (vchange) -- +(1.5, 0) node[midway, below] {$d_{Q}$};
	\draw[normalEdge, blue, dashed, ->] (vchange) -- +(0.2, 1.3) node[right] {$d_{P}^{\prime}$};
	\draw[normalEdge, red, dashed, <-] (vchange) -- +(-0.2, 1.3) node[left] {$d_{Q}^{\prime}$};
	\draw (vchange) ++(-1, 0.5) node[draw] {$-1$};
	\draw (vchange) ++(1, 0.5) node[draw] {$1$};
	
	
	\node[smallNode] (vfour) at (-4, 2.5) {};
	\draw (vfour) +(-2.5, 0) node {(a)};
	\draw[normalEdge, blue, very thick, <-] (vfour) -- +(-1.3, 0.5) node[left] {$d^{P}_{\text{in}}$};
	\draw[normalEdge, red, very thick, <-] (vfour) -- +(1.3, -0.5) node[right] {$d^{Q}_{\text{in}}$};
	\draw[normalEdge, blue, very thick, ->] (vfour) -- +(1.3, 0.5) node[right] {$d^{P}_{\text{out}}$};
	\draw[normalEdge, red, very thick, ->] (vfour) -- +(-1.3, -0.5) node[left] {$d^{Q}_{\text{out}}$};
	\draw (vfour) ++(0, 0.7) node[draw] {$-1$};
	\draw (vfour) ++(0, -0.7) node[draw] {$1$};
	
	
	\node[smallNode] (tdq) at (0.8, 0) {};
	\draw (tdq) +(-1.8, 1.3) node {(c)};
	
	\draw (tdq) +(-1, 1) node[smallNode] (tdp) {} edge[normalEdge, blue, bend right=45, very thick, ->] node[sloped, below] {$d_{P}$} (tdq);
	\draw (tdq) +(1.6, 0) node[smallNode] (hdq) {} edge[normalEdge, red, very thick, <-] node[sloped, below] {$d_{Q}$} (tdq);
	\draw (hdq) +(2, 0) node[smallNode] (tdc) {} edge[normalEdge, very thick, <-] (hdq);
	\draw (tdc) +(1.6, 0) node[smallNode] (hdc) {} edge[normalEdge, red, very thick, <-] node[sloped, below] {$d_{C}$} (tdc);
	\draw (hdc) +(1, 1) node[smallNode] (hdcp) {} edge[normalEdge, red, bend left=45, very thick, <-] node[sloped, below] {$d_{C}^{\prime}$} (hdc)
		edge[normalEdge, bend right=80, very thick, ->] (tdp);
	\draw (tdq) +(1, 1.2) node[smallNode] (hdp) {} edge[normalEdge, blue, dashed, <-]  node[sloped, above] {$d^{\prime}$} (tdq);
	\draw (hdc) +(-1, 1.2) node[smallNode] (tdpp) {} edge[normalEdge, blue, dashed, ->]  node[above, sloped] {$d^{\prime\prime}~~~~~$} (hdc) edge[normalEdge, blue, dashed, <-, bend right=50] (hdp);
	\draw (hdc) +(0.6, -1.2) node[smallNode] {} 
		edge[normalEdge, blue, snake=snake, line before snake=0.4cm, <-, dashed] (hdc)
		node[right] {$t$};
	\draw (hdq) ++(-0.6, 0.5) node[draw] {$1$};
	\draw (tdc) ++(0.4, 0.5) node[draw] {$\geq 2$};
	\draw (hdc) ++(0.25, 0.65) node[draw] {$\geq 1$};
	\draw (hdcp) +(-0.5, 1.5) node {$C$};
	
	
	\draw[gray] (-5, -1.5) rectangle (5, -2.75);
	\draw[normalEdge, ->, blue] (-4.5, -1.75) -- (-3.5, -1.75);
	\node[color=blue, anchor=base] at (-4, -2.2) {$P$-dart};
	\draw[normalEdge, ->, red] (-2.5, -1.75) -- (-1.5, -1.75);
	\node[color=red, anchor=base] at (-2, -2.2) {$Q$-dart};
	\draw[normalEdge, very thick, ->] (-0.5, -1.75) -- (0.5, -1.75);
	\node[anchor=base] at (0, -2.2) {solid dart};
	\draw[normalEdge, dashed, ->] (1.5, -1.75) -- (2.5, -1.75);
	\node[anchor=base] at (2, -2.2) {solid or};
	\node[anchor=base] at (2, -2.6) {hidden dart};
	\node[draw] (phi) at (4,-2) {$\phi(f)$};

\end{tikzpicture}
 \caption[test]{(a) The order at a vertex with four solid darts must be exactly as illustrated, thus paths and cycles of a decomposition of $D^{P \wedge Q}$ cannot cross (cf. Lemma \ref{lem_no_crossing}).\\ (b) A ``change of tracks'': both paths must occur between $d_{Q}$ and $\rev{d_{P}}$ in the counterclockwise order of the incident darts at the particular vertex (Lemma \ref{lem_changeoftracks}).\\ (c) A simplified\footnotemark\ illustration of the situation in the proof of Lemma \ref{lem_cycleconstellation}. When a change of tracks from $P$ to $Q$ occurs on the counterclockwise cycle $C$, $P$ enters the interior. It can only leave again towards the exterior if the potential is at least $2$ as the $Q$-darts must continue the solid cycle.}
 \label{fig_ccwcycle_proof}
 \end{center}
\end{figure}
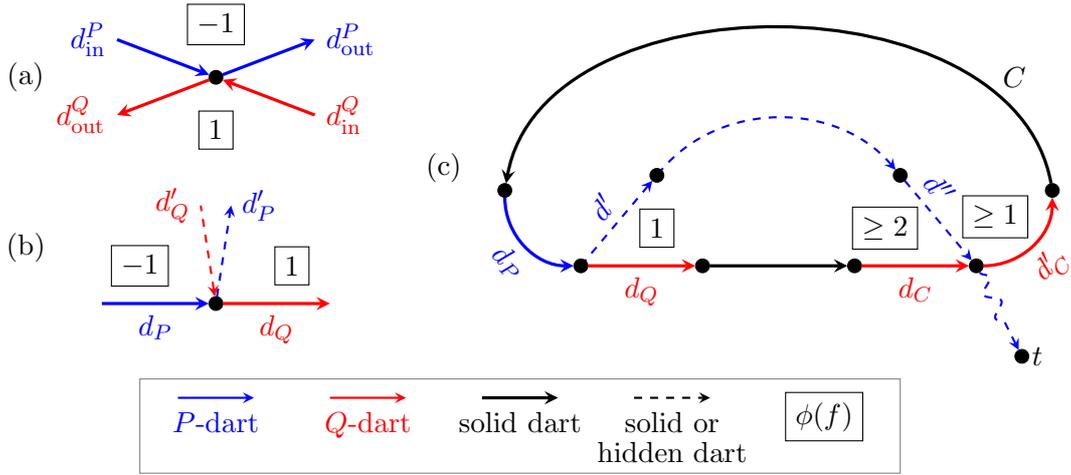\footnotetext{In the illustration, $d_{Q}$ and $d_{C}$ already fulfill the requirements imposed on $d_{1}$ and $d_{2}$ by Lemma \ref{lem_cycleconstellation}. In the general case however, $d_{1}$ and $d_{2}$ can differ from $d_{Q}$ and $d_{C}$, as, e.g., $P[\tail(d_{Q}), \head(d_{C})]$ could meet $C[\tail(d_{Q}), \head(d_{C})]$ at intermediate vertices or even contain darts anti-parallel to those on the cycle.}

We can use Lemma \ref{lem_changeoftracks} to deduce a contradiction from the existence of any solid counterclockwise cycle. The intuitive argument is that a change of tracks from $P$ to $Q$ must happen on the cycle since none of the paths can comprise the cycle on its own, and thus a part of $P$ must leave from that counterclockwise cycle towards its interior. As $t$ is exterior of the cycle (it is at the infinite face), $P$ crosses the cycle again to continue on the exterior. At the vertex at which $P$ crosses the cycle, the potential left of the cycle darts must be greater or equal than $2$ because the subsequent $Q$-darts on the cycle are solid requiring a potential of at least $1$ on their left and the $P$-dart meeting the cycle from the left reduces the potential by $1$ (cf. Figure \ref{fig_ccwcycle_proof} (c)). A contradiction can now be obtained by counting the entering and leaving $P$-darts on the left of the cycle segment. Be aware that the line of argumentation just described is severely simplified and there are several special cases that need to be handled. Thus, we need to introduce a further simplifying assumption before we can proceed.

Suppose we contract an edge with a dart $d \in P \cap Q$. It is easy to check that in the resulting graph both $P \setminus \{d\}$ and $Q \setminus \{d\}$ still are simple $s$-$t$-paths (unless $s$ and $t$ were merged, which implies the trivial case $P = Q$) and that the circulation $\delta_{P\setminus\{d\}} - \delta_{Q\setminus\{d\}}$ in the new graph is represented by the same face potential $\phi$ (the faces are still the same after the contraction, but the edge is deleted in the dual). Thus, the construction performed in Lemma \ref{lem_meetdefinition} produces $D^{P \wedge Q} \setminus \{d\}$ instead of $D^{P \wedge Q}$ for the contracted graph. Finally, if $D^{P \wedge Q}$ contains a counterclockwise simple cycle $C$, then $D^{P \wedge Q} \setminus \{d\}$ contains the counterclockwise simple cycle $C \setminus \{d\}$ in the graph after the contraction (note that $C$ always contains at least two edges that cannot be contracted, as it contains darts of $P \setminus Q$ and $Q \setminus P$).\footnote{The simplicity of $P$, $Q$ is preserved beause $d \in P \cap Q$ implies that a cycle in $P \setminus \{d\}$ or $Q \setminus \{d\}$ that occurs after merging two vertices must have already been a cycle in $P$ or $Q$ before contracting the edge. Simplicity of $C$ is maintained as a dart in $P \cap Q$ connecting two vertices on the cycle can only be a dart of the cycle ($d$ is the only possible way the cycle can continue at $\tail(d)$).} Consequently, we can assume w.l.o.g. that $P \cap Q = \emptyset$. In particular, every  dart on a solid cycle is a $P$-dart or a $Q$-dart.

\begin{assumption}
	For the proofs of Lemma \ref{lem_cycleconstellation} and Lemma \ref{lem_nocounterclockwise}, we assume that $P \cap Q = \emptyset$.
\end{assumption}

\begin{lemma}\label{lem_cycleconstellation}
	Let $C$ be a solid counterclockwise cycle. Then there is a dart $d_{1} \in C$ with $\phi_{l}(d_{1}) = 1$ and a dart $d_{2} \in C$ with $\phi_{l}(d_{2}) \geq 2$ such that $\tail(d_{1})$ and $\head(d_{2})$ are on $P$ and  $C[\tail(d_{1}), \head(d_{2})] \circ \rev{P[\tail(d_{1}), \head(d_{2})]}$ is a simple counterclockwise cycle. 
\begin{proof}	
	$C$ must contain at least one solid $P$-dart $d_{P}$ and one solid $Q$-dart $d_{Q}$, as the paths are simple and none of them can make up the cycle on its own. W.l.o.g., we can assume $\head(d_{P}) = \tail(d_{Q})$ and apply the change of tracks lemma, which implies that $\phi_{l}(d_{Q}) = 1$ and gives us a dart $d^{\prime} \in P$ that leaves $\tail(d_{Q})$ to the left, i.e., towards the interior of $C$. As $P$ ends in $t$, which is neither in the interior nor on the cycle, $P[\tail(d^{\prime}), t]$ contains a dart on the exterior of $C$. Let $d^{\prime\prime}$ be the last dart of $P[\tail(d^{\prime}), t]$ that is in the interior of $C$ before the first exterior dart occurs. As $\head(d^{\prime\prime})$ is on the cycle, it must have an incoming cycle dart and an outgoing cycle dart. As $d^{\prime\prime} \in P$ is not a cylce dart, the incoming cycle dart is a solid $Q$-dart $d_{C}$. If the outgoing cyle dart is a $P$-dart, $P$ continues on the cycle after $d^{\prime\prime}$ until the cycle changes tracks again from a $P$-dart to a $Q$-dart. But then $P$ enters the interior of the cycle again by the change of tracks lemma. This is a contradiction to our choice of $d^{\prime\prime}$. So the outgoing cycle dart is also a solid $Q$-dart $d_{C}^{\prime}$. As there is no other interior $P$-dart at $\head(d^{\prime\prime})$ by choice of $d^{\prime\prime}$, the faces left of $d_{C}$ and $d_{C}^{\prime}$ are only separated by $d^{\prime\prime}$ with $\rightf(d^{\prime\prime}) = \leftf(d_{C})$ and $\leftf(d^{\prime\prime}) = \leftf(d_{C}^{\prime})$. Thus $\phi_{l}(d_{C}) = \phi_{l}(d_{C}^{\prime}) + 1 \geq 2$. (also cf. Figure \ref{fig_ccwcycle_proof} (c))
	
	We now choose $d_{2}$ to be the first dart on $C[\tail(d_{Q}), \head(d_{C})]$ with $\phi_{l}(d_{2}) \geq 2$ and $\head(d_{2})$ on $P[\tail(d_{Q}), t]$. We then choose $d_{1}$ to be the last dart on $C[\tail(d_{Q}), \head(d_{2})]$ with $\tail(d_{1})$ on $P[\tail(d_{Q}),\head(d_{2})]$. Clearly, $d_{1}$ and $d_{2}$ are $Q$-darts. If the predecessor $d_{1}^{\prime}$ of $d_{1}$ on $C$ is a $P$-dart, $\phi_{l}(d_{1}) = 1$ by the change of tracks lemma. Otherwise, $2 > \phi_{l}(d_{1}^{\prime}) > 0$ by choice of $d_{2}$, and there is a $P$-dart entering $\tail(d_{1})$ from the interior by choice of $d_{1}$ and construction of $P[\tail(d_{Q}), t]$, implying $\phi_{l}(d_{1}^{\prime}) \geq \phi_{l}(d_{1})$. Thus, $1 = \phi_{l}(d_{1}^{\prime}) \geq \phi_{l}(d_{1}) > 0$ in this case.
\end{proof}
\end{lemma}

We now apply a counting argument on the $P$-darts that enter or leave $C[\tail(d_{1}), \head(d_{2})]$ from the left to show that $\phi_{l}(d_{2})$ can in fact not be larger than $\phi_{l}(d_{1})$. This yields a contradiction, proving the non-existence of counterclockwise cycles in $D^{P \wedge Q}$.
\newline\newline
\textbf{Lemma \ref{lem_nocounterclockwise}.}\textit{
 There are no counterclockwise simple cycles in $D^{P \wedge Q}$.
}
\begin{proof}		
	Let $d_{1}$ and $d_{2}$ be as asserted by Lemma \ref{lem_cycleconstellation}. We will traverse $C$ from $v := \tail(d_{1})$ to $w := \head(d_{2})$ and show that $\phi_{l}(d_{1}) \geq \phi_{l}(d_{2})$.\footnote{It is actually more convenient to imagine we traverse the interior faces adjacent to the vertices on $C[v, w]$ from $\leftf(d_{1})$ to $\leftf(d_{2})$ in the dual. We then cross every dart that leaves $C[v,w]$ to the left from left to right, decreasing the potential by $1$, and every dart that enters $C[v,w]$ from the left from right to left, increasing the potential by $1$.} Note that $C[v,w]$ starts and ends with a $Q$-dart. Furthermore, whenever a change of tracks on $C[v,w]$ from $Q$ to $P$ and back to $Q$ occurs, the potential to the left must be equal to $1$ at the last $Q$-dart on the cycle before the $P$-darts and at the first $Q$-dart on the cycle after them. Moreover, these two changes correspond to exactly one $P$-dart entering $C[v,w]$ from the left and  exactly one $P$-dart leaving it to the left, with no $P$-darts entering or leaving $C[v,w]$ inbetween. So the total increase of the potential of the face left of us on our way from $d_{1}$ to $d_{2}$ on the cycle is the number of $P$-darts that leave $C[v,w]$ to the left minus the number of $P$-darts that enter $C[v,w]$ from left. However, the number of $P$-darts entering $C[v,w]$ from the left must greater or equal to the number of $P$-darts leaving it to the left because $t$ is on the exterior of $C[v,w] \circ \rev{P[v,w]}$ and $P$ cannot cross $P[v,w]$, such that any part of $P$ that enters the interior from $C[v, w]$ must also leave it later on to $C[v, w]$.\footnote{The vertex $t$ is on the exterior as it is adjacent to the infinite face and can also not be on the cycle as it has no outoing darts.} This implies $1 = \phi_{l}(d_{1}) \geq \phi_{l}(d_{2}) = 2$, a contradiction.
\end{proof}

\subsubsection*{Submodularity of the path lattice}

In the proof of Theorem \ref{thm_leftmostpathlattice} we omitted the submodularity of the lattice. Thus, we still have to show that $(S \wedge T) \cup (S \vee T) \subset S \cup T$ and $(S \wedge T) \cap (S \vee T) \subset S \cap T$ for all $S, T \in \mathcal{P}$. The first inclusion directly follows from the construction of meet and join and Lemma \ref{lem_meetdefinition} as $S \wedge T \subseteq D^{S \wedge T} \subseteq P \cup Q$ and $S \vee T \subseteq D^{S \vee T} \subseteq P \cup Q$. For the second inclusion, let $d \in (P \wedge Q) \cap (P \vee Q) \subseteq D^{P \wedge Q} \cap D^{P \vee Q}$ and assume by contradiction $d \in P \setminus Q$ or $d \in Q \setminus P$. In the first case, Corollary \ref{cor_leftpotential} implies $\delta_{l}(d) < 0$ as $d \in D^{P \wedge Q} \setminus Q$. An analogous version of the corollary for $D^{P \vee Q}$ implies $\delta_{r}(d) > 0$ for $d \in D^{P \vee Q} \setminus Q$. However, $1 = (\delta_{P} - \delta_{Q})(d) = \delta_{r}(d) - \delta_{l}(d) \geq 2$, which is a contradiction. The case $d \in Q \setminus P$ leads to the same contradiction. Thus $d \in P \cap Q$ and the lattice is submodular.

\newpage

\subsection*{Appendix to Section \ref{sec_characterization}: The proof of Lemma \ref{lem_counterexample}}

 \begin{figure}[t]
  \begin{center}
  \begin{tikzpicture}[scale=0.9]
 
	\node[labeledNode] (s) at (-2, 0) {$s$};
	\node[normalNode] (1) at (0, 2) {}
		edge[normalEdge] (s);
	\node[normalNode] (2) at (0, 0) {}
		edge[normalEdge] (s)
		edge[normalEdge] (1);
	\node[normalNode] (3) at (0, -2) {}
		edge[normalEdge] (s)
		edge[normalEdge] (2);
	\node[labeledNode] (t) at (2, 0) {$t$}
		edge[normalEdge] (1)
		edge[normalEdge] (2)
		edge[normalEdge] (3);
	\draw[normalEdge] (0, -2) arc (300:64:2.3);
		
	\path (8, 3) node[draw] (P1) {$P_{1}$};
	\path (P1) ++(-0.75, 0.5) node[smallLabeledNode] (sP1) {} 
		++(0.75,0.75) node[smallNode] (v1P1) {}
		edge[<-, normalEdge] (sP1)
		++(0.75,-0.75) node[smallLabeledNode] (tP1) {}
		edge[<-, normalEdge] (v1P1);
		
	\path (P1) ++(-3.5, -1) node[draw] (P2) {$P_{2}$} edge[->, normalEdge, dashed, bend left=35] node[above] {(3)} (P1);
	\path (P2) ++(0.6, 0) node[smallLabeledNode] (sP2) {} 
		++(0.75, 0) node[smallNode] (v2P2) {}
		edge[<-, normalEdge] (sP2)
		++(0,0.75) node[smallNode] (v1P2) {}
		edge[<-, normalEdge] (v2P2)
		++(0.75,-0.75) node[smallLabeledNode] (tP2) {}
		edge[<-, normalEdge] (v1P2);
		
	\path (P1) ++(3.5, -1) node[draw] (P3) {$P_{3}$} edge [->, normalEdge, dashed, bend right=35] node[above] {(4)} (P1);
	\path (P3) ++(-2.1, 0) node[smallLabeledNode] (sP3) {} 
		++(0.75, 0.75) node[smallNode] (v1P3) {}
		edge[<-, normalEdge] (sP3)
		++(0, -0.75) node[smallNode] (v2P3) {}
		edge[<-, normalEdge] (v1P3)
		++(0.75, 0) node[smallLabeledNode] (tP3) {}
		edge[<-, normalEdge] (v2P3);
		
	\path (P2) ++(-1.2, -2) node[draw] (P4) {$P_{4}$};
	\path (P4) ++(1.1, 0) node[smallLabeledNode] (sP4) {} 
		++(0.75, 0) node[smallNode] (v2P4) {}
		edge[<-, normalEdge] (sP4)
		++(0, -0.75) node[smallNode] (v3P4) {}
		edge[<-, normalEdge] (v2P4)
		++(0, 1.5) node[smallNode] (v1P4) {}
		++(0.75,-0.75) node[smallLabeledNode] (tP4) {}
		edge[<-, normalEdge] (v1P4);
	\draw[->, normalEdge] (v3P4) arc (300:65:0.85);
	
	\path (P3) ++(0.8, -2) node[draw] (P5) {$P_{5}$};
	\path (P5) ++(-2.2, 0) node[smallLabeledNode] (sP5) {} 
		++(0.75, -0.75) node[smallNode] (v3P5) {}
		edge[<-, normalEdge] (sP5)
		++(0, 1.5) node[smallNode] (v1P5) {}
		++(0, -0.7) node[smallNode] (v2P5) {}
		edge[<-, normalEdge] (v1P5)
		++(0.75, 0) node[smallLabeledNode] (tP5) {}
		edge[<-, normalEdge] (v2P5);
	\draw[->, normalEdge] (v3P5) arc (300:65:0.85);
	
	\path (P1) ++(0, -4.6) node[draw] (P6) {$P_{6}$} edge[->, normalEdge, dashed, bend left=35] node[above right] {(5)} (P4) edge[->, normalEdge, dashed, bend right=35] node[above left] {(6)} (P5);
	\path (P6) ++(-0.75, +1.6) node[smallLabeledNode] (sP6) {}
		++(0.75, -0.75) node[smallNode] (v3P6) {}
		edge[<-, normalEdge] (sP6)
		++(0, 1.5) node[smallNode] (v1P6) {}
		++(0.75,-0.75) node[smallLabeledNode] (tP6) {}
		edge[<-, normalEdge] (v1P6);
	\draw[->, normalEdge] (v3P6) arc (300:65:0.85);
	
	\path (P4) ++(0.8, -2.5) node[draw] (P7) {$P_{7}$};
	\path (P7) ++(0.6, 0) node[smallLabeledNode] (sP7) {} 
		++(0.75, 0) node[smallNode] (v2P7) {}
		edge[<-, normalEdge] (sP7)
		++(0,-0.75) node[smallNode] (v3P7) {}
		edge[<-, normalEdge] (v2P7)
		++(0.75, 0.75) node[smallLabeledNode] (tP7) {}
		edge[<-, normalEdge] (v3P7);
		
	\path (P5) ++(-0.4, -2.5) node[draw] (P8) {$P_{8}$};
	\path (P8) ++(-2.1, 0) node[smallLabeledNode] (sP8) {} 
		++(0.75, -0.75) node[smallNode] (v3P8) {}
		edge[<-, normalEdge] (sP8)
		++(0, 0.75) node[smallNode] (v2P8) {}
		edge[<-, normalEdge] (v3P8)
		++(0.75, 0) node[smallLabeledNode] (tP8) {}
		edge[<-, normalEdge] (v2P8);
		
		\draw[->, normalEdge, dashed] (P4) -- node[left] {(1)} (P2);
		\draw[->, normalEdge, dashed] (P7) -- node[left] {(2)} (P4);
		\draw[->, normalEdge, dashed] (P3) -- node[right] {(7)} (P5);
		\draw[->, normalEdge, dashed] (P5) -- node[right] {(8)} (P8);
\end{tikzpicture}
   \caption[test]{The graph $K^{s-t}_{5}$ and eight of its $s$-$t$-paths, which suffice for showing that there is no partial order that induces a submodular and consecutive lattice on the set of paths. The numbered dashed edges correspond to the steps of the proof of Lemma \ref{lem_counterexample}.}
   \label{fig_k5}
  \end{center}
 \end{figure}
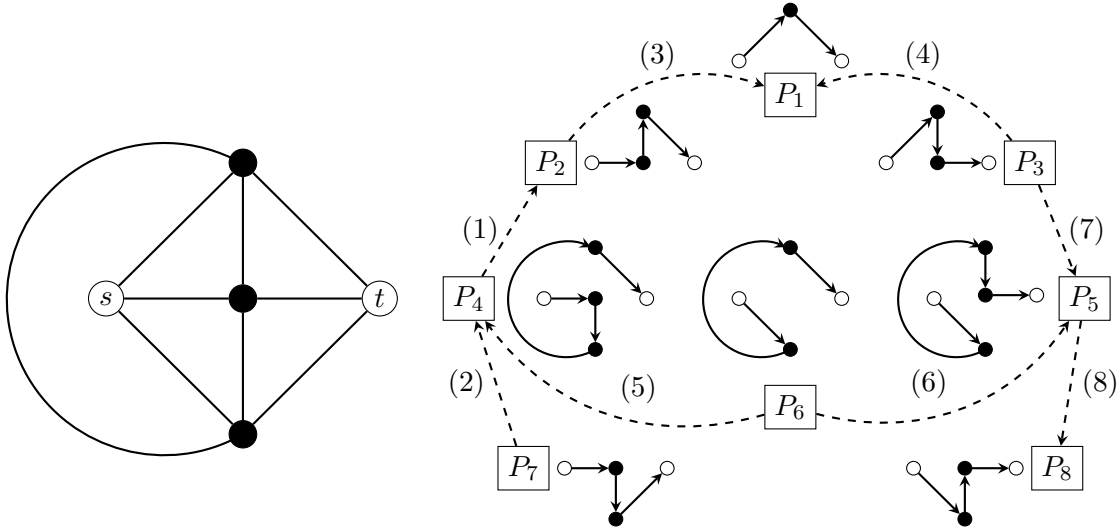

In Section \ref{sec_general}, we stated that the set of $s$-$t$-paths of $K^{s-t}_{3,3}$ or $K^{s-t}_{5}$ cannot be equipped with a consecutive and submodular lattice structure (Lemma \ref{lem_counterexample}). However, we only proved this property for the graph $K^{s-t}_{3,3}$. We now complete the proof for $K^{s-t}_{5}$.
\newline\newline
\textbf{Lemma \ref{lem_counterexample}.}
\textit{
  The set of $s$-$t$-paths in $K^{s-t}_{3,3}$ or $K^{s-t}_{5}$ (or a subdivision of these graphs) cannot be equipped with a partial order $\preceq$, such that $(\mathcal{P}, \preceq)$ is a consecutive and submodular lattice.
}
\begin{proof}
 Consider the eight paths $P_{1}, \ldots, P_{8}$ of $K^{s-t}_{5}$ depicted in Figure \ref{fig_k5}. By Lemma \ref{lem_comparability}, we deduce that $P_{1}$, $P_{2}$, $P_{4}$, and $P_{7}$ are pairwise comparable. W.l.o.g.\ we let $P_{4} \prec P_{2}$ (1). This already implies $P_{7} \prec P_{4}$ (2) and then $P_{2} \prec P_{1}$ (3) by consecutivity. As $P_{3}$ is comparable to $P_{1}$, we obtain $P_{3} \prec P_{1}$ (4) by consecutivity. Furthermore, the pairwise comparability of $P_{2}$, $P_{4}$, and $P_{6}$ implies $P_{6} \prec P_{4}$ (5) by consecutivity. As $P_{6}$ and $P_{5}$ are comparable, we obtain $P_{6} \prec P_{5}$ (6). Also $P_{3}$ and $P_{5}$ are comparable, and as $P_{6} \prec P_{5} \prec P_{3} \prec P_{1}$ contradicts consecutivity, $P_{3} \prec P_{5}$ (7) must hold. This immediately implies $P_{5} \prec P_{8}$ (8) as $P_{3}$, $P_{5}$ and $P_{8}$ must be pairwise comparable. Finally $P_{1}$ and $P_{8}$ must be comparable, but both $P_{4} \prec P_{1} \prec P_{8}$ and $P_{6} \prec P_{8} \prec P_{1}$ yield a contradiction to consecutivity.
\end{proof}
 
 \bibliography{references}{} 
\end{document}